%% file: paper.tex
\documentclass[11pt]{article}
\usepackage[margin=1in]{geometry}
\usepackage{mathpazo}
\usepackage{amsmath,amsthm,amssymb,xspace,verbatim,paralist}
\usepackage{float}
\usepackage{url,hyperref}
\usepackage{tabularx}
\usepackage{algorithm}
\usepackage[noend]{algpseudocode}
\makeatletter
\renewcommand{\ALG@beginalgorithmic}{\footnotesize}
\makeatother
\usepackage{caption}
\usepackage{cite}
\usepackage{cleveref}

\theoremstyle{plain}% default
\newtheorem{theorem}{Theorem}%[section]
\newtheorem{lemma}[theorem]{Lemma}
\newtheorem{obs}[theorem]{Observation}

\theoremstyle{definition}
\newtheorem{definition}[theorem]{Definition}

\title{Submodular Maximization Meets Streaming:\\ Matchings, Matroids, and More}
\author{%
  Amit Chakrabarti%
  \thanks{Department of Computer Science, Dartmouth College. Supported in part by NSF grant CCF-1217375.}
  \and
  Sagar Kale%
  $^\fnsymbol{footnote}$
}

\date{\{ac,\:sag\}@cs.dartmouth.edu}

\input{macros}

\begin{document}

\maketitle

\thispagestyle{empty}
\input{abstract}
\newpage
\addtocounter{page}{-1}

\input{intro}

\input{prelim}

%These section merged
%\input{02sssm}

%\input{03mpsm}
\input{onepass}

\input{multipass}

%Generic alg and its basic properties now inside prelim section
%\input{04gamm}

%no ideas on LB yet.
%\input{05lb}
%\input{05sumfw}

\input{hypergraph}

\input{matroid}

%No success on lbs yet.
%\input{lb}

{\small
  \bibliographystyle{plain}
  \bibliography{references}
}

\appendix

\input{appendix}

\end{document}

%% file: macros.tex
\def\clap#1{\hbox to 0pt{\hss#1\hss}}

\def\mathrlap{\mathpalette\mathrlapinternal}

\def\mathrlapinternal#1#2{\rlap{$\mathsurround=0pt#1{#2}$}}

\newcommand{\argmax}{\operatornamewithlimits{argmax}}
\newcommand{\eps}{\varepsilon}
\newcommand{\e}{\mathrm{e}}

\newcommand{\RR}{\mathbb{R}}
\newcommand{\cA}{\mathcal{A}}

\newcommand{\cX}{\mathcal{X}}
\newcommand{\PP}{\mathsf{P}}
\newcommand{\NP}{\mathsf{NP}}

\renewcommand{\ge}{\geqslant}
\renewcommand{\le}{\leqslant}
\newcommand{\etal}{{\textit{et al.}}\xspace}

\newcommand{\ang}[1]{\langle{#1}\rangle}

\newcommand{\match}[1]{M^{#1}}
\newcommand{\trail}[1]{T^{#1}}
\newcommand{\killed}[1]{K^{#1}}
\newcommand{\ttt}{T}
\newcommand{\kkk}{K}
\newcommand{\both}[1]{B^{#1}}
\newcommand{\factor}[1]{C_{{#1}}}
\newcommand{\OPT}{M^*}
\newcommand{\intedges}{\Rsh}
\newcommand{\prevmatch}{P}
\newcommand{\spcl}{compliant\xspace}
\newcommand{\Spcl}{Compliant\xspace}
\newcommand{\shadow}[2]{\text{shadow-edge}(#1,#2)}

\newcommand{\IC}{\mathcal{I}}

\newcommand{\OPTI}{I^*}
\newcommand{\curv}{\operatorname{curv}}

%% file: abstract.tex
\begin{abstract}

We study the problem of finding a maximum matching in a graph given by an
input stream listing its edges in some arbitrary order, where the quantity to
be maximized is given by a monotone submodular function on subsets of edges.
This problem, which we call maximum submodular-function matching (MSM), is a
natural generalization of maximum weight matching (MWM), which is in turn a
generalization of maximum cardinality matching (MCM).
We give two incomparable algorithms for this problem with space usage falling
in the semi-streaming range---they store only $O(n)$ edges, using $O(n\log n)$
working memory---that achieve approximation ratios of $7.75$ in a single pass
and $(3+\eps)$ in $O(\eps^{-3})$ passes respectively. The operations of these
algorithms mimic those of Zelke's and McGregor's respective algorithms for
MWM; the novelty lies in the analysis for the MSM setting. In fact we identify
a general framework for MWM algorithms that allows this kind of adaptation to
the broader setting of MSM.

In the sequel, we give generalizations of these results where the maximization
is over ``independent sets'' in a very general sense. This generalization 
captures hypermatchings in hypergraphs as well as independence in the 
intersection of multiple matroids.

% To be updated if we manage to prove lower bounds.

\end{abstract}

%% file: intro.tex
\section{Introduction}\label{sec:intro}

Maximum cardinality matchings and maximum weight matchings are basic concepts
in graph theory and efficient algorithms for computing these structures are
fundamental algorithmic results with myriad applications. The
explosion of data---in particular graph data---over the past decade has
motivated a number of researchers to revisit several algorithmic problems on
graphs with a view towards designing {\em space efficient} algorithms that
process their inputs in {\em streaming fashion}, i.e., via sequential
access alone, though perhaps in multiple passes. In particular, a series of
recent works~\cite{fgnbm,mwms,zelke,elms,gkk,Kapralov13} have studied the
maximum cardinality matching (MCM) problem and its natural generalization, the
maximum weight matching (MWM) problem, on graph streams.

In this work, we study a further generalization of MWM that we call the {\em
maximum submodular-function matching} problem or, more briefly, the {\em
maximum submodular matching} (MSM) problem. This specific problem does not
seem to have been studied in previous work, though there has been plenty of
work in the optimization community on general (non-streaming) algorithms for
constrained submodular function maximization under constraints more general
than matchings (see, e.g., Feldman \etal \cite{feld} and the references
therein, as well as our own discussion in \Cref{sec:related}).  Our work gives
the first results for the MSM problem in the data stream model. Our
techniques in fact lead to results for a wider class of problems,
including submodular maximization on hypermatchings and intersection of
matroids.

Our study of MSM is inspired in part by its applicability to the Word
Alignment Problem (WAP) from computational linguistics, as studied in
Lin and Bilmes~\cite{lin}: we are given a source-language string and a
target-language string and the goal is to find a ``good'' mapping
between their respective words.  Lin and Bilmes~\cite{lin} cast WAP as
maximizing a suitable submodular function constrained to the intersection
of two partition matroids (this is in turn closely related to bipartite
matchings), and they note the improvement this gives over previous
approaches that cast WAP as an instance of MWM.

\medskip

A {\em submodular function} on a ground set $\cX$ is defined to be a function
$f:2^{\cX}\to\RR$ that satisfies
%\begin{equation} \label{eq:subm1}
$f(A\cup B)+f(A\cap B) \le f(A)+f(B)$
%\end{equation}
for all $A,B\subseteq \cX$. For our purposes in this work, we will instead use
the following ``diminishing returns'' characterization, which is well known to
be equivalent to the definition just given: for all $Y\subseteq X\subset \cX$ 
and $x \in \cX\setminus X$, we have
\begin{equation} \label{eq:subm2}
  f(X \cup \{x\}) - f(X) \le f(Y \cup \{x\}) - f(Y) \, .
\end{equation}
The function $f$ is said to be {\em monotone} if $f(Y) \le f(X)$ whenever $Y
\subseteq X \subseteq \cX$ and {\em proper} if $f(\emptyset) = 0$.
An instance of MSM consists of a graph $G = (V,E)$ on vertex set $V = [n] :=
\{1,2,\ldots,n\}$ and a non-negative monotone proper submodular function $f$
whose ground set is the edge set $E$, i.e., $f:2^E\to\RR_+$. The goal is to
output a matching $M^* \subseteq E$ that maximizes $f(M^*)$; we shall refer to
such a matching as an $f$-MSM of $G$. For a real number $\alpha\ge 1$, an
$\alpha$-approximate $f$-MSM of $G$ is defined to be a matching $M \subseteq
E$ such that $f(M) \ge \alpha^{-1} f(M^*)$.

A non-negative weight function $w:E\to\RR_+$ can be naturally extended to
subsets of $E$ via $w(S) = \sum_{e\in S} w(e)$ for all $S \subseteq E$; the
latter function $w$ is easily seen to be non-negative, monotone, proper, and
submodular (it is in fact {\em modular}, a.k.a. linear). Therefore MSM
generalizes the more famous MWM problem. Letting $w$ be a constant function
gives us the even more special MCM problem. However, an important threshold is
crossed in generalizing from MWM to MSM. The MWM problem is solvable in
polynomial time~\cite{Edmonds1965a,Galil86-survey}---a monumental algorithmic
triumph of the 20th century---whereas MSM hits an $\Omega(1)$
approximation threshold if it is to be solved in polynomial time;
see \Cref{thm:apx} for a formal treatment.

Our concern in this paper is with graph {\em streams}: the input graph is
described by a stream of edges $\{u,v\}$, with $u,v\in [n]$. We assume that the
number of vertices, $n$, is known in advance and that each edge in $E$ appears
exactly once in the input stream. The order of edge arrivals is arbitrary and
possibly adversarial. We seek algorithms for MSM that use only quasi-linear
working memory---i.e., $O(n(\log n)^{O(1)})$ bits of storage, with $O(n\log
n)$ being the holy grail---and process each edge arrival very quickly, ideally
in $O(1)$ time. Algorithms with such guarantees have come to be known as
{\em semi-streaming} algorithms~\cite{fgnbm}. Notice that $\Omega(n\log n)$
bits are necessary simply to store a matching that saturates $\Omega(n)$
vertices.

As with all optimization problems involving submodular functions, a study of
MSM requires special care because the description of the submodular function
$f$ needs $\Omega(2^{|E|})$ space in general. Special cases do allow $f$ to be
more compactly represented: such as MWM, where each edge in the stream arrives
together with its weight. Since our goal is to give general algorithms,
assuming no further structure for $f$, we will instead take the common
approach of having $f$ specified by a {\em value oracle} that returns $f(S)$
when presented with $S \subseteq E$. See \Cref{sec:prelim} for a precise
explanation.

\subsection{Our Results} \label{sec:results}

We give two incomparable approximation algorithms for the MSM problem on graph
streams, formally stated in the two theorems below. Both algorithms are
semi-streaming: specifically, each stores only $O(n)$ edges, thereby
using $O(n\log n)$ working memory.%
\footnote{Throughout the paper, we adopt the convention that 
edge weights in an MWM instance---and analogously, 
$f$-values of singletons in an MSM instance---do not
grow with $n$; this ensures that each weight we store in our
algorithms takes up $O(1)$ storage.}
For brevity, ``submodular $f$,'' means a non-negative monotone 
proper submodular function $f$, presented by a value oracle.

\begin{theorem} \label{thm:onepass}
  For every submodular $f$, there is a one-pass semi-streaming algorithm that
  outputs a $7.75$-approximate $f$-MSM of an $n$-vertex input graph, storing
  at most $O(n)$ edges at all times.
\end{theorem}

\begin{theorem} \label{thm:multipass}
  For every submodular $f$, and every constant $\eps > 0$, there is a
  multi-pass semi-streaming algorithm that makes $O(\eps^{-3})$ passes over an
  $n$-vertex graph stream and outputs a $(3+\eps)$-approximate $f$-MSM of the
  graph.  This algorithm stores only a matching in the input graph at all
  times; in particular it stores only $O(n)$ edges.
\end{theorem}

Perhaps more important than these specific approximation ratios is the
technique behind these results.  We identify a general framework for matching
algorithms in graph streams.  We show that whenever an MWM algorithm fits this
framework, it can be adapted to the broader setting of MSM. The above theorems
then follow by revisiting two recent MWM algorithms---that of
Zelke~\cite{zelke} for \Cref{thm:onepass} and that of McGregor~\cite{mwms} for
\Cref{thm:multipass}---and showing that they fit our framework.  Thus, the
main contributions of our work are (1) the identification of the framework and
(2) the novel analysis for the MSM setting. 

Naturally, MWM algorithms base their actions on edge weights. A trivial way of
giving a weight to an edge $e$ for the $f$-MSM problem is to use the quantity
$f(\{e\})$. As may be expected, this is too na\"{i}ve to be useful. Our first
insight is that weights can be assigned to edges as they are encountered in
the stream based on how much they improve the ``current matching.'' Our second
insight, specific to multi-pass algorithms, is that edge weights assigned this
way can be calculated on each pass. Though the resulting weights may change
from one pass to another, nevertheless, our framework and analysis technique
allow us to recover a good approximation ratio.

Our framework is not deeply wedded to matchings: it is general enough to
capture set maximization problems constrained to abstract ``independent
sets,'' for a very general notion of independence. Taking this view, we obtain
two more families of results. The first applies to {\em hypergraphs}, where we
obtain approximate MWM and $f$-MSM algorithms for matchings (a.k.a.
hypermatchings) given a bound $p$ on the size of hyperedges. The second
applies to maximization over the intersection of $p$ matroids (for a
constant $p$): the case $p = 2$ captures matchings in bipartite graphs. In
each family, we have a maximum-submodular problem (MSIS, say---the ``IS''
stands for ``independent set''), and a maximum-weight problem (MWIS, say)
where the submodular function is modular. The results are summarized below;
details appear in~\Cref{sec:hypg,sec:matroid}.

\begin{theorem} \label{thm:mwis}
  For every submodular $f$, the MWIS and $f$-MSIS problems, with independent
  sets being given either by a hypermatching constraint in $p$-hypergraphs or
  by the intersection of $p$ matroids, there are near-linear-space
  streaming algorithms giving the following approximation ratios.

  {\renewcommand{\arraystretch}{1.3}
  \begin{center}\begin{tabular}{l|c|c}
    Problem type & MWIS & MSIS \\
    \hline
    One pass:\, $p$-hypergraphs; $p$ matroids & $2(p+\sqrt{p(p-1)})-1$ & $4p$ \\
    $O(\eps^{-3} \log p)$ passes:\, $p$-hypergraphs; $p$ partition matroids & $p+\eps$ & $p+1+\eps$ \\
  \end{tabular}\end{center}}
\end{theorem}

The one-pass MWIS result for matroids was already known from the work of
Ashwinkumar~\cite{ashwin}; the remaining results in~\Cref{thm:mwis} are novel.
%We conjecture that the second row of results hold even without the restriction
%to partition matroids.

Our results for $f$-MSM and $f$-MSIS as stated above involve approximation
ratios that are always worse than the corresponding ones for MWM and MWIS:
therefore they are not generalizations in the strictest sense. However, such a
generalization is possible by considering
the {\em curvature} $\curv(f)$ of the submodular function $f:2^E\to\RR$. This
quantity, defined by
\begin{equation} \label{eq:curv}
  \curv(f) := \min\{c:\, \forall\,A\subseteq E~ \forall\,e\in E\setminus A,
  ~\text{we have}~ f(A\cup\{e\}) - f(A) \ge (1-c) f(\{e\})\} \, ,
\end{equation}
measures how far $f$ is from being modular. Note that $\curv(f) \in [0,1]$ and
$\curv(f) = 0$ iff $f$ is modular. For our final result, we give approximation
ratios for $f$-MSM and $f$-MSIS that gradually improve to those for Zelke's MWM
algorithm and Ashwinkumar's MWIS algorithm as $\curv(f) \to 0$.

\begin{theorem} \label{thm:curvature}
  For every submodular $f$, the approximation ratios for $f$-MSM
  in~\Cref{thm:onepass} and the one-pass approximation ratio for $f$-MSIS
  in~\Cref{thm:mwis} can be improved to $\min\{7.75,\, 5.585/(1-\curv(f))\}$ and
  $\min\{4p,\, (2(p+\sqrt{p(p-1)})-1)/(1-\curv(f))\}$ respectively.
\end{theorem}

\subsection{Context} \label{sec:context}

To place \Cref{thm:onepass,thm:multipass} in context, we summarize the most
relevant prior work on MWM and MCM. As noted before, our work is the first to
consider MSM.

A one-pass semi-streaming $2$-approximation for MCM is trivial, from the
observation that a maximal matching is a $2$-approximate MCM. Beating this
bound or proving its optimality remains a vexing open problem. Goel, Kapralov,
and Khanna~\cite{gkk} showed that for every $\eps > 0$, finding a
$(3/2-\eps)$-approximate MCM in one pass requires $n^{1+\Omega(1/\log\log n)}$
space, and very recently Kapralov~\cite{Kapralov13} improved this to the best
known approximation ratio lower bound of $\e/(\e-1) \approx 1.582$. Using
multiple passes, McGregor~\cite{mwms} did obtain a $(1+\eps)$-approximation
algorithm, but required $\exp(\eps^{-1})$ passes to do so.

For the MWM problem, Feigenbaum \etal \cite{fgnbm} gave a one-pass
semi-streaming $6$-approximation algorithm, which McGregor~\cite{mwms} improved to
$3+\sqrt{8} \approx 5.828$ by tweaking parameters. McGregor also extended his
algorithm to multiple passes, where each pass essentially ``repeats'' the
first pass, gradually improving the approximation ratio to $2+\eps$ after
$O(\eps^{-3})$ passes.  Zelke further improved the one-pass approximation
factor to $\approx 5.585$ (the exact constant is a degree-$5$ algebraic
number) by using a more involved algorithm. We shall have reason to discuss
Zelke's and McGregor's algorithms in detail in
\Cref{sec:onepass,sec:multipass} respectively. Most recently, Epstein \etal
\cite{elms} gave the current best approximation ratio of $\approx 4.911 +
\eps$, using $O(n\log(n/\eps))$ space.  Their algorithm departs significantly
from previous approaches, and falls outside our aforementioned framework. 
% It has the slight drawback of requiring $O(\log(n/\eps))$ time to process each
% edge arrival, whereas previous approaches used only $O(1)$ time.
Very recently, Ahn and Guha~\cite{ahnguha} gave a
$(1+\eps)$ approximation with multiple passes, using a very different
algorithm that falls outside our framework.

It is also worth noting another line of work on MWM that has focused on {\em
bipartite} graphs, where it is natural to consider an alternate streaming
model in which {\em vertices} arrive together with all their incident edges.
The seminal online (randomized) algorithm of Karp, Vazirani, and
Vazirani~\cite{KarpVV90} falls in this setting and gives an $\e/(\e-1)$
competitive ratio.  Recent work of Goel, Kapralov, and Khanna~\cite{gkk} gave
a deterministic semi-streaming algorithm with the same approximation ratio.
Kapralov's aforementioned lower bound~\cite{Kapralov13}, which holds despite
the bipartite and vertex-arrival restrictions, shows that this ratio is
optimal.

\subsection{Other Related Work} \label{sec:related}

Thus far, we have been thinking about MSM as a problem about finding a
good matching, generalizing MWM. Another useful viewpoint is to consider
$f$-MSM as the problem of maximizing the submodular function $f(S)$ subject to
$S$ being a matching. This makes MSM an instance of constrained submodular
maximization, which is a heavily-studied topic in optimization.

Maximizing a submodular function is an important problem even without
constraints on the set if one drops the monotonicity requirement. Even its
most famous instance, {\sc Max-Cut}, is not yet fully understood. For monotone
functions, maximizing $f(S)$ becomes nontrivial once one places some sort of
``packing'' constraint on $S$, such as an upper bound on $|S|$. Generalizing
this idea naturally leads one to a {\em matroid constraint}, where $S$ is
required to be an independent set of a matroid. One can consider more general
``independence systems,'' such as the intersection of $p$ different matroids
on the same ground set $E$ (called a $p$-intersection system) or, even more
generally, a {\em $p$-system}, wherein 
\vspace{-9pt}
\[
  \forall\, A \subseteq E:~
    \frac{\max\{|S|:\, S \subseteq A,\, S~\text{maximally independent}\}}
      {\min\{|S|:\, S \subseteq A,\, S~\text{maximally independent}\}} \le p \, .
\]
This last generalization finally captures the constraint of $S$ being a
matching, because matchings form a $2$-system.  All of these classes of
problems were studied in the seminal work of Fisher, Nemhauser and
Wolsey~\cite{nwf1,nwf2}, who showed among other things that the simple greedy
strategy of growing a set $S$ by adding the element that most improves $f(S)$
subject to $S$ being independent yields a $(p+1)$-approximation for a
$p$-system. In another classical work, Jenkyns~\cite{jenkyns}
showed that the greedy strategy with a $p$-system constraint in fact gives a
$p$-approximate solution if the function $f$ is modular (a.k.a. linear).

Notice that the implication of the above for our $f$-MSM problem is a
(non-streaming) greedy $3$-approximation. This should be compared to
our \Cref{thm:multipass}.

More recently, Calinescu \etal \cite{CCPV} gave a polynomial time
$(\e/(\e-1))$-approximation algorithm for maximizing a non-negative monotone
proper submodular function $f$ subject to a matroid constraint. This ratio
improves upon the $2$ that follows from a matroid being a $1$-system, and is
provably optimal if $\PP \ne \NP$. Lee \etal \cite{lmns,lsv} gave {\em local
search} algorithms for maximizing $f$ over a $p$-intersection system,
improving the approximation ratio from the aforementioned $p+1$ (for the more
general $p$-systems) to $p+\eps$. Recently, Feldman \etal \cite{feld} proposed
a new class of independence systems called ``$p$-exchange systems,'' stricter
than $p$-systems but more general than $p$-intersection systems, for which
they gave a local-search-based $(p+\eps)$-approximation. Their paper is highly
recommended for a concise yet comprehensive summary of relevant work on
submodular maximization. Matchings form a $2$-exchange system. Therefore this
last result improved---after a span of over 30 years---the best known
approximation ratio for $f$-MSM from $3$~\cite{nwf2} to $2+\eps$. We invite the 
reader to
compare again with our \Cref{thm:multipass}.

%Concurrent with our work, Ashwinkumar and Vondr\'{a}k~\cite{ashjan} gave a 
%$(p+1+\eps)$-approximation algorithm for submodular maximization for a 
%$p$-system that uses $O((m/\eps^2)\log^2 (m/\eps))$ $f$-oracle calls. This can 
%be compared with our $(p+1+\eps)$-approximation algorithm that uses 
%$O(m/\eps^3)$ $f$-oracle calls (assuming $p$ to be a constant) for $p$ 
%partition matroids (recall that $p$-systems are more general). 

\subsection{Motivation and Significance of Our Results}

In applications such as big data analytics, it is sometimes preferable to
compute a good solution {\em quickly}, even if a theoretically stronger
guarantee can be achieved by a slower algorithm.  Our algorithms in this work
should be seen in this light: they are significant because they are {\em
faster} algorithms with slightly worse approximation ratios than best known
offline approximation algorithms.  Moreover, they are able to handle input
presented in streaming fashion, a clear advantage when handling big data.

Notably, none of the algorithmic strategies discussed in~\Cref{sec:related}
are suitable for use with a graph {\em stream}. Let us focus just on our
problem, $f$-MSM. For an $m$-edge graph, a typical step of the greedy strategy
requires an examination of $\Omega(m)$ edges and $\Omega(m)$ calls to the
$f$-oracle. The local search algorithms need to find a ``good'' local move in
each step and in general it is not clear how a single pass over a graph stream
can guarantee more than one such good move. So a local search algorithm that
makes $\tau$ moves potentially translates into a $\tau$-pass streaming
algorithm, and the best upper bound for $\tau$ in the aforementioned
$(2+\eps)$-approximation algorithm appears to be $O(\eps^{-1} n^4)$. In
contrast, our algorithms make $O(1)$ calls to the $f$-oracle per input item
per pass, use a constant number of passes, and use an essentially optimal
amount of storage.

Very recently, and concurrent with our work, Ashwinkumar and
Vondr\'{a}k~\cite{ashjan} gave a $(p+1+\eps)$-approximation algorithm for
submodular maximization over a $p$-system. Their algorithm can be thought of
as using $O(\eps^{-2} \log^2(m/\eps))$ passes, where $m=|E|$. Our result in
Theorem~\ref{thm:mwis} uses fewer passes (for constant $\eps$), but handles a
smaller class of constraints than $p$-systems.

As noted at the start, the generalization from MWM to MSM has some practical
motivation, such as for the WAP problem from computational
linguistics~\cite{lin}. In addition, we feel that the MSM problem is a
pleasing marriage of submodular maximization, data streaming, and matching
theory; and more generally, MSIS brings together submodularity, streaming, and
matroids.

%% file: prelim.tex
\section{Preliminaries} \label{sec:prelim}

We start by making our model of computation precise. The input is an
$n$-vertex graph stream, defined as a sequence $\sigma =
\ang{e_1,e_2,\ldots,e_m}$ of distinct edges, where each $e_i = (u_i, v_i) \in
[n] \times [n]$ and $u_i < v_i$. We put $V = [n], E = \{e_1,\ldots,e_m\}$, and
$G = (V,E)$. The submodular function $f:2^E\to\RR_+$, which is part of the
problem specification, is given by an entity external to the stream, called
the {\em value oracle} for $f$, or the $f$-oracle. A data stream algorithm,
after reading each edge from the input stream, is allowed to make an arbitrary
number of {\em calls} to the $f$-oracle. (In fact the algorithms we design
here make only $O(1)$ such calls on average.) A call consists of the algorithm
sending the oracle a subset $S \subseteq E$, whereupon the oracle returns the
value $f(S)$ in constant time.  The space required to describe $S$ counts
towards the algorithm's space usage.

Notice that $f$ is only defined on subsets of $E$, and the most important
restriction on the algorithm is that at any time it can only remember a tiny
portion of $E$. To prevent the algorithm from ``cheating'' and learning about
$E$ indirectly from oracle calls, we say that the algorithm {\em fails} or
{\em aborts} if it ever tries to obtain $f(S)$ with $S \not\subseteq E$.

\subsection{Hardness of MSM}

We cannot hope to solve MSM exactly. As noted in \Cref{sec:context}, even the
very special case MCM cannot be approximated any better than $\e/(\e-1)$ in
the semi-streaming setting~\cite{Kapralov13}. But the data stream model does
not adequately capture the vast gulf between MWM and MSM.

\begin{theorem} \label{thm:apx}
  For every $C < \e/(\e-1)$, there does not exist a polynomial-time
  $C$-approximation algorithm for $f$-MSM relative to a value oracle for $f$.
\end{theorem}
\begin{proof}
  Given a submodular $f$, we can view $f$-MSM as a generalization of the
  constrained maximization problem $\max_{|S|\le k} f(S)$: consider
  $f$-MSM on a disjoint
  union of $k$ star graphs. Nemhauser and Wolsey~\cite[Theorem~4.2]{nemwol}
  show that for the latter problem, given the upper bound on $C$, a
  $C$-approximation algorithm must make a superpolynomial number of calls to
  the $f$-oracle.
\end{proof}

\subsection{A Framework for Streaming MSM and MSIS Algorithms} \label{sec:generic}

We proceed to describe a generic streaming algorithm for $f$-MSM, which
defines the framework alluded to in \Cref{sec:results}. In fact, as noted
towards the end of~\Cref{sec:results}, our framework applies to the much more
general problem of $f$-MSIS (Maximum Submodular Independent Set), an instance
of which is given by a submodular $f:2^E\to\RR_+$ and a collection
$\IC\subseteq 2^E$ of {\em independent sets} such that $\emptyset\in\IC$. We
put $m := |E|$ and $n := \max_{I\in\IC} |I|$. We assume that independence
(i.e., membership in $\IC$) can be tested easily; we require no other structural
property of $\IC$. For the special case of MSM,
$\IC$ is the collection of matchings in a graph with edge set $E$.

The generic algorithm for $f$-MSIS starts with a given independent set
$\prevmatch$ (possibly empty) and then proceeds to make {\em one} pass over
the input stream $\sigma$, attempting to end up with an improved independent
set $I$ by the end of the pass. The algorithm processes the elements in $E$ in
a {\em pretend stream order} that consists of an arbitrary permutation of the
elements in $\prevmatch$, followed by the elements in $E \setminus \prevmatch$
in the same order as $\sigma$.
Throughout, the algorithm maintains a ``current solution'' $I\in\IC$, a set
$S\subseteq E$ of ``shadow elements'' (this term is borrowed from
Zelke~\cite{zelke}), and a weight $w(e)$ for each element $e$ it has
processed.  The algorithm bases its decisions on a real-valued parameter
$\gamma > 0$. For a set $A \subseteq E$, we denote $w(A) := \sum_{e\in A}
w(e)$. An {\em augmenting pair} for a set $I\in\IC$ is a pair of sets $(A,J)$
such that
%$A\cap I=\emptyset$, 
$J \subseteq I$ and $(I \setminus J)\cup A \in
\IC$. For $e\in E$, define $A+e$ to be $A\cup\{e\}$.

\algrenewcommand\algorithmicforall{\textbf{foreach}}
\begin{algorithm}[!ht]
  \caption{~~Generic One-Pass Independent Set Improvement Algorithm for $f$-MSIS
    \label{alg:generic}}
  \begin{algorithmic}[1]
    \Function{Improve-Solution}{$\sigma, \prevmatch, \gamma$}
    \State $I \gets \emptyset,\, S \gets \emptyset$
    \ForAll{$e \in \prevmatch$ in some arbitrary order}
	%\Call{Process-Element}{$e,I,S$} \vspace{-2pt} 
	%\label{alg:generic:initloop}
	$w(e) \gets f(I + e) - f(I)$, $I\gets I + e$ 
	\vspace{-2pt} \label{alg:generic:initloop}
    \EndFor
    \ForAll{$e \in \sigma\setminus \prevmatch$ in the order given by $\sigma$}
	\Call{Process-Element}{$e,I,S$} \vspace{-2pt} \label{alg:generic:mainloop}
    \EndFor
    \State \Return{$I$}
    \EndFunction
    \Statex

    \Procedure{Process-Element}{$e,I,S$} \Comment{{\footnotesize Note: Assigns weight 
    $w(e)$ and modifies $I$ and $S$.}}
    \State $w(e) \gets f(I \cup S + e) - f(I \cup S)$ \label{alg:generic:wdef}
    \State $(A, J)\gets$ a well-chosen augmenting pair for $I$ with 
    $A \subseteq I\cup S + e$ and $w(A) \ge (1+\gamma)w(J)$ \label{alg:generic:improve}
    \State $S \gets$ a well-chosen subset of $(S \setminus A) \cup J$ 
    \label{alg:generic:updateS}
	\State $I \gets (I \setminus J) \cup A$ \Comment{{\footnotesize Augment 
	independent set $I$ using $A$.}}
	\label{alg:generic:updateM}
    \EndProcedure
    \end{algorithmic}
\end{algorithm}

%%% A new element is always first taken into $I$ and whenever removed from $I$ it can 
%%% be moved to $S$. A shadow element can move back to being an independent set 
%%% element and vice 
%%% versa. But if an element is removed from $S$, then it has no chance of coming 
%%% back.

Notice that $\Call{Process-Element}{}$ maintains the invariant that $w(e)$ is
defined for all $e\in I\cup S$. Therefore, \Cref{alg:generic:improve} never 
tries to
access an element weight before defining it. Furthermore, the algorithm need
only remember the weights of elements in $I \cup S$. Therefore, the space
usage of the algorithm is bounded by $O((|P| + |I| + |S|)\log m) = O((n +
|S|)\log m)$, since $P$ and $I$ are independent sets.

To instantiate this generic algorithm, one must specify the precise logic used
in \Cref{alg:generic:improve,alg:generic:updateS}. If the algorithm is for
MWIS rather than MSIS, then $w(e)$ values are already given and assignments to
those values (see \Cref{alg:generic:initloop,alg:generic:wdef}) should be ignored.

\begin{definition} \label{def:spcl}
  We say that an MWIS algorithm is {\em \spcl} if each pass instantiates
  \Cref{alg:generic} in the above sense, i.e., it starts with some solution
  $\prevmatch\in\IC$ computed in the previous pass and calls
  $\Call{Improve-Solution}{\sigma,\prevmatch,\gamma}$.  The parameter $\gamma$
  need not be the same for all passes.
\end{definition}

\begin{definition}
  For a submodular $f$, we define an {\em $f$-extension} of a \spcl MWIS
  algorithm $\cA$ to be \Cref{alg:generic}, with the logic used in
  \Crefrange{alg:generic:improve}{alg:generic:updateS} being borrowed from
  $\cA$, and with values of the parameter $\gamma$ possibly differing from
  those used by $\cA$.  \end{definition}

\begin{lemma}[Modular to submodular] \label{lem:generic}
  Let $\mathcal{A}$ be a one-pass \spcl MWIS algorithm that computes a
  $\factor{\gamma}$-approximate MWIS when run with parameter $\gamma$. Then,
  for every non-negative monotone proper submodular $f$, its $f$-extension
  with parameter $\gamma$
  computes a $(\factor{\gamma} + 1 + 1/\gamma)$-approximate $f$-MSIS.
\end{lemma}

The rest of the paper is organized as follows. {\bf \Cref{sec:onepass}}
develops some important properties of \spcl algorithms and proves
\Cref{lem:generic}. Applying this lemma, the $f$-extension of Zelke's one-pass
MWM algorithm~\cite{zelke} is easily shown to yield a $7.75$-approximation for
$f$-MSM, proving \Cref{thm:onepass}.  {\bf \Cref{sec:multipass}} revisits
McGregor's multi-pass algorithm and analysis~\cite{mwms}, extends his
analysis, and obtains a $(3+\eps)$-approximation for $f$-MSM, proving
\Cref{thm:multipass}. {\bf \Cref{sec:hypg,sec:matroid}} briefly discuss
our results for MWM and MSM in hypergraphs, and MWIS and MSIS for the
intersection of $p$ matroids. Technical details of these results are given in
the appendix.

%% file: onepass.tex
\section{A One-Pass Solution via \Spcl Algorithms} \label{sec:onepass}

Consider the $f$-extension with parameter $\gamma$ of a particular one-pass
compliant algorithm. Let $I$ denote its output, i.e., the result of invoking
$\Call{Improve-Solution}{\sigma,\emptyset,\gamma}$ and $\OPTI$ be an $f$-MSIS.
Let $I_e, S_e$ denote the contents of the variables $I,S$ in
\Cref{alg:generic} just before element $e$ is processed.  Let
$\kkk=(\bigcup_{e\in E}I_e)\setminus I$ denote the set of elements that were
added to the current solution at some point but were {\em killed} and did not make it to
the final output.  Then $\bigcup_{e\in E}S_e \subseteq I\cup \kkk$, because an
element can become a shadow element only when it was removed from the
current solution at some point (see \Cref{alg:generic:updateS} of
\Cref{alg:generic}). Hence
\begin{equation}
  \textstyle \bigcup_{e\in E} (I_e\cup S_e) \subseteq I\cup \kkk \, .
  \label{eq:sh}
\end{equation}

\begin{lemma}\label{lem:btrailg}
   For an $f$-extension of a \spcl algorithm, we have $w(\kkk)\le w(I)/\gamma$.
\end{lemma}
\begin{proof}
  Let $A_e, J_e$ be the sets $A, J$ chosen at \Cref{alg:generic:improve} when
  processing $e$. Each augmentation by $A_e$ (\Cref{alg:generic:updateM})
  increases the weight of the current solution by $w(A_e)-w(J_e) \ge 
  \gamma w(J_e)$.  Hence,
  $w(I)/\gamma \ge \textstyle \sum_{e\in E}w(J_e)$.

  The set $\bigcup_{e\in E} J_e$ consists of elements that were
  removed from the current solution at some point. Thus, it includes 
  $K$
  (the inclusion may be proper: $\kkk$ does not contain elements that were 
  removed from the current solution, reinserted, and eventually ended up in 
  $I$).  Therefore,
  \[
    w(\kkk)
    \le w\big(\textstyle \bigcup_{e\in E} J_e\big)
    \le \sum_{e\in E}w(J_e)
    \le w(I)/\gamma \, . \qedhere
  \]
\end{proof}

\begin{lemma}\label{lem:wmbg}
  For an $f$-extension of a \spcl algorithm, we have $w(I)\le f(I)$.
\end{lemma}
\begin{proof}
  Let $e^I_1, e^I_2,\ldots, e^I_s$ be an enumeration of $I$ in order of
  processing, where $s = |I|$. The logic of \Cref{alg:generic} ensures that an
  element once removed from the shadow set can never return to the current
  solution (though elements can move
  between the two arbitrarily). Thus, $I\cap
  (I_{e^I_i}\cup S_{e^I_i})=\{e^I_1, e^I_2,\ldots, e^I_{i-1}\}$. 
  % because an output element must always be a maintained independent set 
  % element or a shadow element.  
  Since $I\cap (I_{e^I_i}\cup S_{e^I_i}) \subseteq (I_{e^I_i}\cup S_{e^I_i})$
  and $f$ is submodular, \Cref{eq:subm2} gives
  \[f(\{e^I_1, e^I_2,\ldots, e^I_i\})-f(\{e^I_1, e^I_2,\ldots,e^I_{i-1}\})
    \ge f(I_{e^I_i}\cup S_{e^I_i}+e^I_i)-f(I_{e^I_i}\cup S_{e^I_i})
    =w(e^I_i)\ .\]
  Summing this over $i\in[s]$ gives
  $f(I) = f(I)-f(\emptyset)\ge \sum_{i=1}^s w(e^I_i)=w(I)$.
\end{proof}

\begin{lemma} \label{lem:obg}
  For an $f$-extension of a \spcl algorithm, we have
  $f(\OPTI)\le (1/\gamma + 1) f(I) + w(\OPTI)$.
\end{lemma}
\begin{proof}
  Let $e^B_1, \ldots, e^B_b$ be an enumeration of $B := I\cup \kkk$ in
  order of processing.  The set $B_i := \{e^B_1, \ldots, e^B_{i-1}\}$ consists
  of elements inserted into the current solution before $e_i^B$ was 
  processed.
  Meanwhile $I_{e^B_i}\cup S_{e^B_i}$ is the subset of these elements that were
  not removed before $e^B_i$ was processed.  Thus, $B_i\supseteq I_{e^B_i}\cup
  S_{e^B_i}$ for all $i\in [b]$.  By submodularity of~$f$ and \Cref{eq:subm2},
  \[
    f(\{e^B_1, e^B_2,\ldots, e^B_i\})-f(\{e^B_1, e^B_2,\ldots,
    e^B_{i-1}\})\le f(I_{e^B_i}\cup S_{e^B_i}+e^B_i)-f(I_{e^B_i}\cup 
    S_{e^B_i})=w(e^B_i) \, .
  \]
  Summing this over $i\in[b]$ gives
  $f(B) = f(B)-f(\emptyset)\le w(B)$.
  Thus, we have
  \begin{equation}
    f(I\cup \kkk)
    \le w(I\cup \kkk)
    =   w(I)+w(\kkk)
    \le f(I)+w(I)/\gamma
    %\le f(I)+\frac{f(I)}{\gamma}
    =   (1/\gamma + 1) f(I)\ ,
    \label{eq:mutb}
  \end{equation}
  where the last two inequalities use \Cref{lem:wmbg} and \Cref{lem:btrailg}
  respectively.

  Now we bound $f(\OPTI)$.
  Let $\OPTI \setminus (I\cup 
  \kkk)=\{e^{\OPTI}_1,e^{\OPTI}_2,\ldots,e^{\OPTI}_t\}$; this enumeration is in 
  arbitrary order. Put $D_0=I\cup \kkk$, $D_i=I\cup \kkk \cup \{e^{\OPTI}_1, 
  \ldots, e^{\OPTI}_i\}$ for $i\in[t]$.
  % hence, $D_t=I\cup \kkk\cup \OPTI$.
  % Now, for all $i\in \{1,\ldots,t\}$, we have 
  By \Cref{eq:sh}, $D_{i-1}\supseteq I\cup \kkk \supseteq I_{e^{\OPTI}_i}\cup 
  S_{e^{\OPTI}_i}$.
  Appealing to submodularity and \Cref{eq:subm2} again,
  % We keep on adding elements from the set $\OPTI \setminus (I\cup \kkk)$, 
  % which are
  % $\{e^{\OPTI}_1,e^{\OPTI}_2,\ldots,e^{\OPTI}_t\}$, to $I\cup \kkk$, and by
  % submodularity of $f$, the increase we get will be bounded by weights of 
  % those
  % elements, that is, for all $i\in \{1,\ldots,t\}$,
  \[
    f(D_i)-f(D_{i-1})
    \le f(I_{e^{\OPTI}_i}\cup 
    S_{e^{\OPTI}_i}+e^{\OPTI}_i)-f(I_{e^{\OPTI}_i}\cup S_{e^{\OPTI}_i}) = 
    w(e^{\OPTI}_i) \, .
  \]
  Summing this over $i\in[t]$ gives
  $f(D_t)-f(D_0)\le w(\OPTI\setminus (I\cup \kkk))\le w(\OPTI)$.
  In other words, $f(I\cup \kkk\cup \OPTI)-f(I\cup \kkk)\le w(\OPTI)$. By 
  monotonicity of
  $f$ and \Cref{eq:mutb}, we have
  \begin{equation}
    f(\OPTI)\le f(I\cup \kkk\cup \OPTI)\le f(I\cup \kkk) + w(\OPTI)
    \le (1/\gamma + 1) f(I) + w(\OPTI)\ .
    \label{eq:owtb}
    \qedhere
  \end{equation}
\end{proof}

\begin{proof}[\textbf{Proof of \Cref{lem:generic}}]
  Since the \spcl algorithm $\cA$ outputs a $\factor{\gamma}$-approximate MWIS,
  it satisfies $w(\OPTI)\le \factor{\gamma} w(I)$ for any weight assignment; in
  particular, the weights assigned by its $f$-extension. Using
  \Cref{lem:obg} and \Cref{lem:wmbg}, we conclude that $f(\OPTI) \le
  (\factor{\gamma} + 1 + 1/\gamma) f(I)$.
\end{proof}

\begin{proof}[\textbf{Proof of \Cref{thm:onepass}}]
  Recall that an ``independent set'' is just a matching in the setting of MWM
  and MSM. Zelke's algorithm chooses the augmenting pair $(A,J)$ as follows:
  $A$ is chosen from an $O(1)$-sized ``neighborhood'' of the edge $e$ being
  processed, and $J$ is set to be $M\intedges A$: the set of edges in $M$ that
  share a vertex with some edge in $A$. It chooses $S$ so that each shadow
  edge intersects some edge in the current matching, thus enforcing $|S| =
  O(n)$ and a space bound of $O(n\log n)$ bits. For the reader's convenience
  we spell out the logic of the algorithm in full in \Cref{app:zelke}.

  Zelke's algorithm is \spcl with $\factor{\gamma} = 2(1+\gamma) +
  (1/\gamma+1) - \gamma/(1+\gamma)^2$ ~\cite[Theorem~3]{zelke}.  By
  \Cref{lem:generic}, its $f$-extension yields an approximation ratio of
  $2(1+\gamma)^2/\gamma - \gamma/(1+\gamma)^2$, which attains a minimum value
  of $7.75$ at $\gamma=1$. This proves the theorem.
\end{proof}

\subsection{Approximation Ratio in Terms of Curvature}

We now show how to obtain the stronger guarantee for $f$-MSM given
in~\Cref{thm:curvature}. The tool we need is the following strengthening
of~\Cref{lem:generic}.

\begin{lemma} \label{lem:curv}
  Let $\cA$ be a one-pass \spcl MWIS algorithm that computes a
  $\factor{\gamma}$-approximate MWIS when run with parameter $\gamma$. Then,
  for every non-negative monotone proper submodular $f$, its $f$-extension
  with parameter $\gamma$
  computes a $\min\{\factor{\gamma}+1+1/\gamma, 
  \factor{\gamma}/(1-\curv(f))\}$-approximate $f$-MSIS.
\end{lemma}
\begin{proof}
  We can bound $f(\OPTI)$ as follows.
  \begin{align*}
    f(\OPTI)
    &\le \sum_{e\in\OPTI} f(\{e\}) &&\text{by submodularity of $f$}\\
    &\le\frac{1}{1-\curv(f)}\sum_{e\in\OPTI}(f(I_e\cup S_e +e)-f(I_e\cup S_e))
      &&\text{by definition of curvature}\\
    &=\frac{w(\OPTI)}{1-\curv(f)} &&\\
    &\le\frac{w(I)\factor{\gamma}}{1-\curv(f)}
      &&\text{by $\factor{\gamma}$-approximation guarantee of $\cA$}\\
    &\le\frac{f(I)\factor{\gamma}}{1-\curv(f)}
      &&\text{by \Cref{lem:wmbg}.}
  \end{align*}
  Thus the $f$-extension of $\cA$ achieves an approximation ratio of at most
  $\factor{\gamma}/(1-\curv(f))$. Combining this with~\Cref{lem:generic} completes 
  the proof.
\end{proof}

\begin{proof}[\textbf{Proof of \Cref{thm:curvature}}]
  We appeal to \Cref{lem:curv}. 
  Recall the expression for $\factor{\gamma}$ for Zelke's algorithm, given in
  the proof of~\Cref{thm:onepass} above.  For $f$-MSM, the claimed
  approximation ratio follows by picking the better of the two solutions
  obtained by running two $f$-extensions of Zelke's algorithm in parallel: one
  with $\gamma = 0.717$, which minimizes $\factor{\gamma}$, and another with
  $\gamma = 1$, which minimizes $\factor{\gamma}+1+1/\gamma$. 

  A similar idea applied to the appropriate compliant algorithms (outlined
  in~\Cref{sec:hypg,sec:matroid}) gives the claimed results for $f$-MSIS on
  $p$-hypergraphs and $p$-intersection systems.

  Note that, in all of these cases, we can avoid having to run two parallel
  $f$-extensions if we knew $\curv(f)$ in advance, for we could then simply
  figure out which value of $\gamma$ gives the better approximation ratio. For
  instance, in the case of $f$-MSM, we would pick $\gamma = 1$ if
  $5.585/(1-\curv(f)) \ge 7.75$ and $\gamma = 0.717$ otherwise.
\end{proof}

%% file: multipass.tex
\section{A Multi-Pass MSM Algorithm} \label{sec:multipass}

In this section we prove \Cref{thm:multipass}.  For this we first review
McGregor's multi-pass MWM algorithm~\cite{mwms}, which is \spcl.  Our
algorithm is simply its $f$-extension, as explained in \Cref{sec:generic}.

To describe McGregor's algorithm with respect to our framework
(\Cref{alg:generic}), we need only explain the two choices made inside
\Call{Process-Edge}{}. These are especially simple. The algorithm never
creates any shadow edges, so \Cref{alg:generic:updateS} always chooses $S =
\emptyset$. In \Cref{alg:generic:improve}, the augmenting pair $(A,J)$ is
chosen so that $A = \{e\}$ if possible, and $A = \emptyset$ otherwise, and $J
= M \intedges A$.  Recall that $M \intedges A$ denotes the set of edges in
matching $M$ that share a vertex with some edge in set $A$.  This describes a
single pass. The overall algorithm starts with an empty matching and
repeatedly invokes $\Call{Improve-Matching}{}$ with $\gamma = 1/\sqrt2$ for
the first pass and $\gamma = 2\eps/3$ for the remaining passes. It stops when
the multiplicative improvement made in a pass drops below a certain
well-chosen rational function of $\gamma$. McGregor analyzes this algorithm to
show that it makes at most $O(\eps^{-3})$ passes and terminates with a
$(2+2\eps)$-approximate MWM.

In our $f$-extension, we make the following tweaks to the parameter $\gamma$:
we use $\gamma = 1$ for the first pass and $\gamma = \eps/3$ for the remaining
passes.  For the reader's convenience, we lay out the logic of the resulting
$f$-MSM algorithm explicitly in \Cref{alg:multipass}. The function
$\Call{Improve-Matching}{}$ is exactly as in \Cref{alg:generic} except that it
calls $\Call{Process-Edge}{e,M}$, since $S$ is never used. 

\begin{algorithm}[!ht]
  \caption{~~Multi-Pass Algorithm for $f$-MSM \label{alg:multipass}}
  \begin{algorithmic}[1]
    \Function{Multi-Pass-MSM}{$\sigma$}
    \State $M\gets$ \Call{Improve-Matching}{$\sigma,\emptyset,1$}
	\label{alg:multipass:first}
	\Comment{{\footnotesize See \Cref{alg:generic}. Obtains $8$-approximate 
      $f$-MSM.}}
    \State $\gamma \gets \eps/3,~
	\kappa \gets {\gamma^3}/(2+3\gamma+\gamma^2-\gamma^3)$ 
	\label{alg:multipass:kdef}
    \Repeat
	\State $w_{\mathrm{prev}} \gets f(M)$
	\State $M\gets \Call{Improve-Matching}{\sigma,M,\gamma}$ 
	\label{alg:multipass:impsol}
	\Until{$w(M)/w_{\mathrm{prev}}\le 1+\kappa$} \label{alg:multipass:stop}
    \State \Return{$M$}
    \EndFunction
    \Statex
    \Procedure{Process-Edge}{$e,M$}
	\Comment{{\footnotesize Compare with \Cref{alg:generic}.}}
    \State $w(e) \gets f(M + e) - f(M)$ \label{alg:multipass:wdef}
    \If {$w(e) \ge (1+\gamma)w(M \intedges \{e\})$} 
    \label{alg:multipass:choose-edge}
	\State $M \gets M \setminus (M \intedges \{e\}) + e$ 
	\label{alg:multipass:updateM}
    \EndIf
    \EndProcedure
    \end{algorithmic}
\end{algorithm} 
% \mnote{\textsc{Improve-Matching} probably belong here, not in earlier 
% section}

Let $\match{i}$ denote the matching $M$ computed by \Cref{alg:multipass} at
the end of its $i$th pass over $\sigma$.  When an edge $e$ is added to $M$ in
\Cref{alg:multipass:updateM}, we say that $e$ is {\em born} and that it {\em
kills} the (at most two) edges in $M \intedges \{e\}$. Notice that during pass
$i > 1$, thanks to the {\em pretend stream order} in which edges are
processed, initially all edges in $\match{i-1}$ are born without killing
anybody%
\footnote{This subtlety appears to have been missed in McGregor's
analysis~\cite{mwms} and it creates a gap in his argument. Using a pretend
stream order as we do in this work fixes that gap.}
(cf.~the discussion at the start of \Cref{sec:generic}); for the rest of the
pass these edges are never considered for addition to $M$.

Let $\killed{i}$ denote the set of edges killed during pass $i$ (some of them
may be born during a subsequent pass). Then $\match{i} \cup \killed{i}$ is
exactly the set of edges born in pass $i$. These edges can be made the
nodes of a collection of disjoint rooted {\em killing trees}%
\footnote{Feigenbaum \etal \cite{fgnbm} and McGregor \cite{mwms} used the
evocative term ``trail of the dead'' for this concept.}
where the parent of a killed edge $e$ is the edge $e'$ that killed it. The
set of roots of these killing trees is precisely $\match{i}$. Let
$\trail{i}(e)$ denote the set of strict descendants of $e \in \match{i}$ in
its killing tree. Then $\killed{i} = \bigcup_{e\in\match{i}} \trail{i}(e)$.

Let $\both{i} = \match{i} \cap \match{i-1}$ denote the set of edges that pass
$i$ retains in the matching from the previous pass. By the preceding
discussion, it follows that $\trail{i}(e) = \emptyset$ for all $e \in
\both{i}$.

\subsection{Analysis}

We now analyze \Cref{alg:multipass}. As before, let $\OPT$ denote an optimal
solution to the $f$-MSM instance. We first prove an approximation guarantee
for the first pass. It is not the best possible one-pass result (see
\Cref{thm:onepass}), but an $O(1)$-approximation suffices, so we can use the
simpler algorithm.

\begin{lemma} \label{lem:multipass-first}
  The matching $\match{1}$ is an $8$-approximate $f$-MSM, i.e., $f(\match{1})
  \ge f(\OPT)/8$.
\end{lemma}
\begin{proof}
  The first pass of the algorithm is a one-pass \spcl algorithm. As shown by
  McGregor~\cite[Lemma 3]{mwms}, its approximation factor is $\factor{\gamma}
  = 1/\gamma + 3 + 2\gamma$. Applying \Cref{lem:generic}, we have
  $f(\OPT)/f(\match{1}) \le 2/\gamma + 4 + 2\gamma$. This bound is minimized
  at $\gamma = 1$ (explaining the choice made in \Cref{alg:multipass:first})
  at which point it evaluates to $8$.
\end{proof}

Define $\tau$ to be the number of passes made by \Cref{alg:multipass}. Let
$w_i(e)$ denote the weight assigned to edge $e$ in \Cref{alg:multipass:wdef}
during the $i$th pass. For the rest of this section, $\gamma$ denotes the
parameter value used by passes $2$ through $\tau$, and $\kappa$ denotes the
corresponding value assigned at \Cref{alg:multipass:kdef}. To analyze the
result of those passes, we first borrow three results---stated in the next
three lemmas---from McGregor's analysis~\cite[Lemma~3 and Theorem~3]{mwms},
which in turn borrows from the Feigenbaum \etal
analysis~\cite[Theorem~2]{fgnbm}.

\begin{lemma}\label{lem:btrailmp}
  For all $i \in [2,\tau]$ and all $e\in \match{i}$, we have
  $w_i(\trail{i}(e)) \le w_i(e)/\gamma$.
\end{lemma}
\begin{proof}
  Directly analogous to \Cref{lem:btrailg}.
  % This is a consequence of the geometric increase in weights up the levels of
  % a killing tree, which is enforced by \Cref{alg:multipass:choose-edge}.
  % We give a full proof in the appendix.
\end{proof}

\begin{lemma} \label{lem:lbbi}
  We have $w_\tau(\both{\tau}) / w_\tau(\match{\tau}) \ge (\gamma-\kappa) / 
  (\gamma+\gamma\kappa)$.
\end{lemma}
\begin{proof}[Proof sketch]
  The logic in \Crefrange{alg:multipass:choose-edge}{alg:multipass:updateM}
  ensures that, for all $i \in [2,\tau]$, we have $w_i(\match{i} \setminus
  \both{i}) \ge (1+\gamma) w_i(\match{i-1} \setminus \both{i})$. In
  particular, this inequality holds at $i = \tau$.

  During the initial phase of pass $i \ge 2$, the set $M$ is {\em 
  monotonically}
  built up from $\emptyset$ to $\match{i-1}$ according to a pretend stream
  order and weights are assigned to edges in $\match{i-1}$ according to
  \Cref{alg:multipass:wdef}. Because of this monotonicity, summing the weights
  of these edges causes the $f$ terms to telescope, giving $w_i(\match{i-1}) =  
  f(\match{i-1})$.
  So the stopping criterion in \Cref{alg:multipass:stop} ensures that
  $w_\tau(\match{\tau}) / w_\tau(\match{\tau-1}) \le 1+\kappa$. Combining this
  with the inequality in the last paragraph (at $i = \tau$) yields the lemma 
  after some straightforward algebra.
\end{proof}

\begin{lemma} \label{lem:charging}
  For all $i \in [2,\tau]$, we have $w_i(\OPT) \le (1+\gamma) \sum_{e\in 
    \match{i}} (w_i(\trail{i}(e)) + 2 w_i(e))$.
\end{lemma}
\begin{proof}[Proof sketch]
  This lemma has a rather creative proof, wherein the weights of edges in
  $\OPT$ are {\em charged} to edges in $\match{i} \cup \killed{i}$ using a
  careful charge transfer scheme.  For the sake of completeness we give a full 
  proof in \Cref{app:charging}.
\end{proof}

We are now ready to fully analyze the approximation guarantee and complexity
of \Cref{alg:multipass}, thereby proving \Cref{thm:multipass}.

\begin{proof}[\textbf{Proof of \Cref{thm:multipass}}]
  As noted earlier, $\trail{i}(e) = \emptyset$ for all $e \in \both{i}$ and 
  $\killed{i} = \bigcup_{e\in \match{i}} \trail{i}(e)$.
  Therefore, $\killed{\tau} = \bigcup_{e\in \match{\tau} \setminus \both{\tau}}
  \trail{\tau}(e)$, which gives
  \begin{equation}
    w_\tau(\killed{\tau})
    =   \sum_{e\in \match{\tau} \setminus \both{\tau}} w_\tau(\trail{\tau}(e))
    \le \sum_{e\in \match{\tau} \setminus \both{\tau}} \frac{w_\tau(e)}{\gamma}
    =   \frac{w_\tau(\match{\tau} \setminus \both{\tau})}{\gamma}
    =   \frac{w_\tau(\match{\tau})- w_\tau(\both{\tau})}{\gamma} \, ,
    \label{eq:mp1}
  \end{equation}
  where the inequality follows from \Cref{lem:btrailmp}.  Using \Cref{eq:mp1} 
  and relating the first and third terms in \Cref{eq:mutb}, we get
  \begin{equation}
    f(\match{\tau}\cup \killed{\tau})
    \le w_\tau(\match{\tau}) + 
    \frac{w_\tau(\match{\tau})-w_\tau(\both{\tau})}{\gamma}
    =   \left(1+\frac{1}{\gamma}\right)w_\tau(\match{\tau})
	-\frac{1}{\gamma} w_\tau(\both{\tau}) \, .
    \label{eq:mp2}
  \end{equation}
  Using ~\Cref{lem:charging}, we now get
  \begin{align*}
    w_\tau(\OPT)
    &\le (1+\gamma) \sum_{e\in \match{\tau}} (w_\tau(\trail{\tau}(e)) + 2 
    w_\tau(e)) && \\
    &=   (1+\gamma)\left[\left(\sum_{e\in \match{\tau} \setminus \both{\tau}}
	 w_\tau(\trail{\tau}(e))\right)
	 + 2 w_\tau(\match{\tau})\right] &&\text{since $\forall e\in 
	   \both{\tau}$, we have $\trail{\tau}(e)=\emptyset$}\\
    &\le (1+\gamma)\left[\left(
	 \sum_{e\in \match{\tau} \setminus \both{\tau}}
	 \frac{w_\tau(e)}{\gamma}\right)
	 + 2 w_\tau(\match{\tau})\right]
	 &&\text{using \Cref{lem:btrailmp}}\\
    &=   \frac{1+\gamma}{\gamma}(w_\tau(\match{\tau})-w_\tau(\both{\tau}))+
	 (2+2\gamma) w_\tau(\match{\tau})\\
    &=   \left(\frac{1}{\gamma}+3+2\gamma\right) w_\tau(\match{\tau})-
	 \left(1+\frac{1}{\gamma}\right) w_\tau(\both{\tau}) \, .
  \end{align*}
  By using \Cref{eq:owtb}, we have
  $f(\OPT)\le f(\match{\tau}\cup \killed{\tau}) + w_\tau(\OPT)$. So
  using \Cref{eq:mp2} we get
  \begin{align*}
    f(\OPT)
    &\le \left(1+\frac{1}{\gamma}\right)w_\tau(\match{\tau})
	 - \frac{1}{\gamma} w_\tau(\both{\tau})
	 + \left(\frac{1}{\gamma}+3+2\gamma\right) 
	 \mathrlap{w_\tau(\match{\tau})
	 - \left(1+\frac{1}{\gamma}\right) w_\tau(\both{\tau})}\\
    &=   
    w_\tau(\match{\tau})\left[1+\frac{1}{\gamma}+\frac{1}{\gamma}+3+2\gamma\right]
	 - w_\tau(\both{\tau})\left[1+\frac{2}{\gamma}\right]\\
    &\le \left[4+\frac{2}{\gamma}+2\gamma - \left(1+\frac{2}{\gamma}\right)
	 \frac{\gamma-\kappa}{\gamma+\gamma\kappa} \right]w_\tau(\match{\tau}) 
	 && \text{using \Cref{lem:lbbi}}\\
    &=   (3+3\gamma) w_\tau(\match{\tau}) && \text{substituting $\kappa 
      =\frac{\gamma^3}{2+3\gamma+\gamma^2-\gamma^3}$},\\
    &\le (3+\eps)f(\match{\tau}) && \text{using \Cref{lem:wmbg}}
  \end{align*}
  and this completes the proof that \Cref{alg:multipass} computes a
  $(3+\eps)$-approximate $f$-MSM.

  Finally we bound the number of passes made by the algorithm. This requires
  some care. It is not as immediate as the corresponding analysis in
  McGregor's MWM algorithm, because the weight function $w_i$ changes from
  pass to pass. We proceed as follows.

  For all $i \in [2, \tau-1]$, we have 
  % $\frac{w_i(\match{i})}{w_i(\match{i-1})}>(1+\kappa)$, i.e.,
  \[
    f(\match{i})
    \ge w_i(\match{i})              %&&\text{using~\Cref{lem:wmbg}},\\
    > (1+\kappa)f(\match{i-1}) \, , %&&\text{using \Cref{eq:wt2}}.\\
  \]
  where the first inequality uses \Cref{lem:wmbg}, and the second uses the
  stopping criterion in \Cref{alg:multipass:stop}. Since $\OPT$ is optimal,
  $f(\OPT)\ge f(\match{i})$, and repeated application of the above inequality
  gives us
  \[
    f(\OPT)
    \ge f(\match{i})
    > (1+\kappa)^{i-1} f(\match{1})
    \ge (1+\kappa)^{i-1} f(\OPT)/8 \, ,
  \]
  where the final step uses \Cref{lem:multipass-first}. Applying this at $i =
  \tau-1$ gives $\tau \le 2 + \log_{1+\kappa} 8 = O(\kappa^{-1}) =
  O(\gamma^{-3}) = O(\eps^{-3})$. Thus, the algorithm finishes in
  $O(\eps^{-3})$ passes, as claimed.
\end{proof}

%% file: hypergraph.tex
\section{Generalization to Matchings in Hypergraphs} \label{sec:hypg}

A hypergraph is a pair $H=(V,E)$, where $V$ is a finite set and $E$ is a
collection of subsets of $V$. It is a $p$-hypergraph if $|e| \le p$ for all
$e\in E$.  A matching in $H$ is a pairwise disjoint subcollection of $E$.
McGregor's one-pass and multi-pass algorithms for MWM and its $f$-extensions
we gave above can be generalized to compute approximate MWMs and $f$-MSMs in
$p$-hypergraphs. For MWM in $p$-hypergraphs, we get approximation ratios of
$2(p+\sqrt{p(p-1)})-1$ in one pass and $p+\eps$ with $O(\eps^{-3}\log p)$ 
passes.
For MSM, the respective approximation ratios we obtain are $4p$ and
$p+1+\eps$.
We define $n:=|V|$ and $m:=|E|\le n^p$, so the space usage is $O(n\log n)$.

To obtain these results, we generalize the charging scheme
alluded to in the proof of \Cref{lem:charging}. We argue that weights
of edges in $\OPT$ can be charged and these charges redistributed such that
\begin{equation}
  w(\OPT)\leqslant \sum_{e\in M}(1+\gamma)\left(
  (p-1)w(T(e))+pw(e)\right)\ ,
  \label{eq:redebh}
\end{equation}
where $T(e)$ is the set of non-roots in the killing tree of $e$.
This argument appears in \Cref{app:charging}.

The rest of the analysis is identical in all four cases $\{$MWM, MSM$\} \times
\{$one-pass, multi-pass$\}$.
In the multi-pass algorithms, we have to set values of $\gamma$ and $\kappa$.
We use $\gamma=\eps/(p+1)$ for both MSM and MWM.
We use $\kappa= \gamma^3/((p-1)(1+\gamma)^2-\gamma^3)$ for MWM and
$\kappa= \gamma^3/(p+(2p-1)\gamma +(p-1)\gamma^2-\gamma^3)$ for MSM
to get the desired approximation ratios for \Cref{thm:mwis}.

To get better approximation ratios in terms of $\curv(f)$, as stated in
\Cref{thm:curvature}, we again appeal to \Cref{lem:curv}, applying it to the
compliant one-pass hypergraph MWM algorithm. We can then use the same idea as
in \Cref{sec:onepass}, taking the better of two values of $\gamma$, to obtain
a one-pass approximation ratio of
$\min\{4p,(2(p+\sqrt{p(p-1)})-1)/(1-\curv(f))\}$.

%% file: matroid.tex
\section{Maximization Over (Multiple) Matroids} \label{sec:matroid}

A \emph{matroid} is a pair $M=(E,\IC)$ such that $E$
is a finite set, $\IC$ is a collection of subsets of $E$, and the following
conditions hold:
\begin{inparaenum}[(1)\,]
  \item $\emptyset \in \IC$;~
  \item if $I\in \IC$ and $J\subseteq I$, then $J\in \IC$ (in
    other words, $\IC$ is closed under the subset operation);~
  \item if $I,J \in \IC$ and $|J|<|I|$, then there exists an element
    $e\in I-J$ such that $J+e\in \IC$.
\end{inparaenum}
A set $I\in \IC$ is also called an \emph{independent set}. A maximally 
independent set is a \emph{base} of the matroid. A set that is not independent 
is \emph{dependent}. A minimally dependent set is a \emph{circuit}.
Since the number of independent sets in a matroid can be exponential, we assume 
that access to $\IC$ is via an oracle. When a set $I$ is passed to this oracle, 
it returns an empty set if $I$ is independent, otherwise it returns a circuit 
in $I$.

Given $p$ matroids $M_1=(E,\IC_1),\ldots,M_p=(E,\IC_p)$ over the same ground 
set $E$, and a nonnegative monotone proper submodular function $f:2^E\to\RR_+$, 
we consider the problem of finding $\argmax_{I\in\IC_1\cap\cdots\cap\IC_p} 
f(I)$, in the streaming model, where elements of $E$ arrive in the stream. We 
define $m:=|E|$ and $n:=\max_{I\in \IC_1\cap\cdots\cap\IC_p}|I|$. Then our 
algorithms use $O(n(\log m)^{O(1)})$ memory.

Ashwinkumar~\cite{ashwin} gave a one-pass
$(2(p+\sqrt{p(p-1)})-1)$-approximation algorithm (and a matching lower bound
when the algorithm is only allowed to store a feasible solution) when $f$ is
modular. His algorithm, which we explain in \Cref{app:ashwin}, is \spcl.
Hence, by \Cref{lem:generic}, we get an approximation ratio of $4p$ when $f$
is submodular (by setting $\gamma=1$).  Now, we seek to improve this ratio 
using multiple passes.  It
is not clear to us if we can extend Ashwinkumar's charging scheme to argue
that his algorithm can be used in \Cref{alg:multipass:impsol} of
\Cref{alg:multipass}. So we give a simpler charging scheme for partition
matroids in \Cref{app:ashwin} to give better approximation ratios using
multiple passes. We get the approximation ratios $p+\eps$ and $p+1+\eps$ for
the modular and submodular cases, respectively, by setting $\gamma=\eps/(p+1)$
for both cases, and $\kappa= \gamma^3/((p-1)(1+\gamma)^2-\gamma^3)$ for the
modular case and $\kappa= \gamma^3/(p+(2p-1)\gamma +(p-1)\gamma^2-\gamma^3)$
for the submodular case, in \Cref{alg:multipass}. 

The better, curvature-dependent approximation ratio of \Cref{thm:curvature}
can be obtained using the same idea as for $f$-MSM, by appealing to
\Cref{lem:curv} and applying it to Ashwinkumar's one-pass algorithm, which, as
we have noted, is compliant.

%% file: appendix.tex
\section{Details of the Charging Schemes}\label{app:charging}
We give a self-contained proof of \Cref{lem:charging} for the sake of
completeness. We also give a generalization of the charging scheme for 
hypergraphs.

We use the same terminology and notation as in previous sections, but we also 
denote an edge $\{u,v\}\in E$ as $uv$, for convenience. Matching $\OPT$ can be 
thought of as any matching, not just an optimum matching, for this argument.
We can charge the weight $w(\OPT)$ to the edges in $M\cup \kkk$ in such a way 
that each edge $e\in \kkk$ gets at most $(1+\gamma)w(e)$ charge and each edge 
$e\in M$ gets at most $2(1+\gamma)w(e)$ charge. This is essentially the same 
charging and charge-redistribution scheme as in Feigenbaum \etal \cite{fgnbm}.

Let $\OPT=\{o_1p_1,o_2p_2,\ldots\}$. Consider an edge $op\in \OPT$. If $op$ was 
born and killed later by an edge $pq$, then we charge $w(op)$ to $pq$ and 
associate this charge with vertex $p$. If $op$ was born and not killed, then it 
charges itself $w(op)$. If $op$ was not born because of one or two edges 
$xo,py\in M\cup \kkk$ then we charge $\frac{w(op)w(xo)}{w(xo)+w(py)}$ to $xo$ 
and associate this charge with vertex $o$ and $\frac{w(op)w(py)}{w(xo)+w(py)}$ 
to $py$ and associate this charge with vertex $p$.
In all three cases the following holds.
\begin{obs}
  \label{obs:1}
  At most $(1+\gamma) w(xy)$ charge is associated with each of $x$ and $y$ for 
  an edge $xy\in M\cup \kkk$.
\end{obs}

The redistribution is done as follows. Imagine this happening in the order the 
edges were processed: whenever an edge $xy$ is killed by an edge $yz$, then we 
transfer the charge associated with vertex $y$ from edge $xy$ to $yz$. Since 
$yz$ killed $xy$, we have $w(yz)\geqslant w(xy)$, hence using 
%Cref won't work below.
Observation~\ref{obs:1}, $yz$ is charged at most $(1+\gamma) w(xy)$ and this 
charge is associated with vertex $y$ after redistribution, which is at most 
$(1+\gamma) w(yz)$.

A charge associated with a vertex $v$ is from an edge, say $uv\in \OPT$. Hence, 
after this redistribution, each edge in $\kkk$ has charge associated with at 
most one vertex, and each edge in $M$ has charge associated with at most two 
vertices, and this completes the proof.

\subsection{Charging Scheme for Hypergraphs}

The charging and redistribution argument is generalized as follows. Recall that 
an edge $e$ in a hypergraph is a subset of the set of vertices, i.e., if $e\in 
E$, then $e\subseteq V$.

Here, we charge the weight $w(\OPT)$ to the edges in $M\cup \kkk$ such that 
each edge $e\in \kkk$ gets at most $(p-1)(1+\gamma)w(e)$ charge and each edge 
$e\in M$ gets at most $p(1+\gamma)w(e)$ charge.
Consider an edge $e^*\in \OPT$. If $e^*$ was born and killed later, then we 
charge $w(e^*)$ to its murderer $e^\dagger$ and associate this charge with 
vertices in $e^*\cap e^\dagger$. If $e^*$ was born and not killed, then it 
charges itself.  If $e^*$ is not born because of at most $p$ edges, say 
$e_1,\ldots, e_p$, then we charge $w(e_i)w(e^*)/\sum_i w(e_i)$ to edge $e_i$ 
and associate these charges with vertices $e^*\cap e_i$ correspondingly. In all 
the cases, an edge $e$ is charged at most $(1+\gamma)w(e)$ by a given edge in 
$\OPT$, and thus bears at most $p(1+\gamma)$ charge.

For redistribution, when an edge $e'$ is killed by $e^k$, we transfer all the 
charge associated with vertices $e'\cap e^k$ to $e^k$. For example, if an edge 
$e=\{a,b,c,d\}$ did not let edges $\{a,b\}$ and $\{c\}$ be born, then $e$ bears 
$w(\{a,b\})$ charge associated with $a$ and $b$, and $w(\{c\})$ charge 
associated with $c$. If in the future $\{a\}$ kills $e$, then we transfer 
$w(\{a,b\})$ 
charge to $\{a\}$, or if $\{a,c,x,y\}$ kills $e$, then we transfer 
$w(\{a,b\})+w(\{c\})$ charge from $e$ to $\{a,c,x,y\}$. Thus, an edge $e$ not 
in the final matching bears at most $(p-1)(1+\gamma)w(e)$ charge, and an edge 
$e$ in the final matching bears at most $p(1+\gamma)w(e)$ charge.  This gives 
us \Cref{eq:redebh}.

\section{Ashwinkumar's Algorithm and a New Charging Scheme}\label{app:ashwin}
We give Ashwinkumar's Algorithm for $p$ matroids and a simpler charging scheme 
for partition matroids. 

Let $(E^1,\ldots,E^r)$ be a partition of $E$, and $k^1,\ldots,k^r$ be positive 
integers, and
\[
  \IC=\{I\subseteq E : |I\cap E^j|\le k^j\text{ for all }j\in[r]\}\ ,
\]
then $(E,\IC)$ is a matroid, specifically, a \emph{partition matroid}. Also, a set 
that is not independent is \emph{dependent}, and a minimally dependent set is a 
\emph{circuit}. We say that a set $J\subseteq E$ \emph{saturates} $E^s$ if 
$|J\cap E^s|= k^s$.
We denote the partition of matroid $M_i$ by $E_i^1,\ldots$, and corresponding 
constraints by $k_i^1,\ldots$.
For an independent set $I\in \IC_i$ and an element $e\in E-I$ such that 
$I+e\notin \IC_i$, let $C_i(I+e)$ denote the unique circuit in $I+e$. Now, 
consider \Cref{alg:pmatroids}.

\algrenewcommand\algorithmicforall{\textbf{foreach}}
\begin{algorithm}[!ht]
  \caption{~~One-Pass Algorithm for Maximization Over Multiple 
    Matroids\label{alg:pmatroids}}
  \begin{algorithmic}[1]
    \Function{Find-Max-Weight-Ind-Set}{}
    \State $I \gets \emptyset$
    \ForAll{$e \in E$ in the order given by $\sigma$}
	\State $(e_1,\ldots,e_p) \gets$ smallest weight elements in 
	$C_1(I+e),\ldots,C_p(I+e)$ \label{alg:pmatroids:intelts}
	\If{$w(e)\ge (1+\gamma) w(\{e_1,\ldots,e_p\})$}
	\State $I \gets (I\setminus\{e_1,\ldots,e_p\})\cup\{e\}$
	\label{alg:pmatroids:updateI}
	\EndIf
    \EndFor
    \State \Return{$I$}
    \EndFunction
  \end{algorithmic}
\end{algorithm}

\subsection{A Simpler Charging Scheme For Partition Matroids}
Let $\OPTI$ be an optimal independent set and $I$ be the output. We reuse some 
of the terminology used in \Cref{sec:multipass}. When an element $e$ is added 
to $I$ in \Cref{alg:pmatroids:updateI}, we say that $e$ is {\em born} and that 
it {\em kills} the (at most $p$) elements $e_1,\ldots,e_p$. Let $I_e$ denote 
the maintained independent set just before $e$ was processed. Let 
$\kkk=(\bigcup_{e\in E}I_e)\setminus I$ denote the set of elements that were 
added to the maintained independent set at some point but did not make it to 
the final output. These elements can be made the nodes of a collection of 
disjoint rooted {\em killing trees} where the parent of a killed element $e$ is 
the edge $e'$ that killed it. The set of roots of these killing trees is 
precisely $I$. Let $\ttt(e)$ denote the set of strict descendants of $e \in I$ 
in its killing tree. Then $\kkk = \bigcup_{e\in I} \ttt(e)$.

We give a charging scheme that is based on the one mentioned in the proof of 
\Cref{lem:charging} but is slightly different from that used in case of 
hypergraphs. We charge responsible elements with respect to a matroid, i.e., if 
$e\in\OPTI$ charges $e'\in I_e$ because $I_e+e$ formed a circuit in some $M_i$, 
then we associate this charge to $e'$ with $M_i$. Now, if $e\in\OPTI$ is taken, 
then it charges itself and this charge is associated with all $p$ matroids.  
Suppose $e\in \OPTI$ was not taken because of the elements 
$e_1^{j_1},\ldots,e_p^{j_p}$ (see \Cref{alg:pmatroids:intelts}), where 
$j_1,\ldots,j_p$ denote the set indices in partitions of $M_1,\ldots,M_p$. Then 
for all $e'\in ((C_1(I_e+e)-e)\cup\cdots\cup(C_p(I_e+e)-e))$, we have 
$w(e)<(1+\gamma) w(e')$.  Also, $e\in E_i^{j_i}$ for $i\in[p]$ and some 
corresponding $j_i$.  Since $I_e$ saturates $E_i^{j_i}$ for all $i\in[p]$, 
there exist elements $e_1'\in (C_1(I_e+e)-e),\ldots,e_p'\in (C_p(I_e+e)-e)$, 
such that $e_1',\ldots,e_p'$ were uncharged with respect to $M_1,\ldots,M_p$, 
respectively, because for all $i\in[p]$, we have $|\OPTI\cap E_i^{j_i}|\le 
k_i^{j_i}$. So, for $i\in[p]$, we charge $w(e_i')w(e)/(\sum_{j=1}^p w(e_j'))$ 
to $e_i'$. In any case, an element $e'$ is charged at most $(1+\gamma)w(e')$ 
and by at most one optimal element with respect to a particular matroid, i.e., 
$p(1+\gamma)w(e')$ in total.

The charge redistribution is as follows. Whenever elements 
$e_1^{j_1},\ldots,e_p^{j_p}$ are killed by $e^k$, we transfer the charge on 
$e_1^{j_1},\ldots,e_p^{j_p}$ associated with $M_1,\ldots,M_p$, respectively, to 
$e^k$ keeping the same association, and the amount of charge transferred is at 
most $(1+\gamma)w(e^k)$. Note that after this transfer, if an $e^k$ carried a 
charge associated with some $M_i$, then it will not be charged again for that 
matroid.  Thus, after all the redistribution, an element $e'\in\kkk$ carries at 
most $(p-1)(1+\gamma)w(e')$ charge, and an element $e'\in I$ carries at most 
$p(1+\gamma)w(e')$ charge. So we get an inequality similar to \Cref{eq:redebh} 
for hypergraphs,
\begin{equation}
  w(\OPTI)\leqslant \sum_{e\in I}(1+\gamma)\left(
  (p-1)w(\ttt(e))+pw(e)\right)\ .
  \label{eq:redebpm}
\end{equation}
%Using the proof of \Cref{lem:btrailg}, we have $\sum_{e\in 
%I}w(\ttt(e))=w(\kkk)\le w(I)/\gamma$, which we substitute in \Cref{eq:redebpm} 
%to get an approximation ratio of $(1+\gamma)((p-1)/\gamma + p)$, which is 
%minimized at $\gamma=\sqrt{p/(p-1)}$ to $2(p+\sqrt{p(p-1)})-1$.  Now, 
%\Cref{alg:pmatroids} is \spcl with $\factor{\gamma} = (1+\gamma)((p-1)/\gamma 
%+ p)$ if we set $A=\{e\}$, $I\intedges A=\{e_1,\ldots,e_p\}$ in 
%\Cref{alg:pmatroids}. By \Cref{lem:generic}, its $f$-extension yields an 
%approximation ratio of $2p+p/\gamma+p\gamma$, which attains a minimum value of 
%$4p$ at $\gamma=1$.

Now, \Cref{eq:redebpm} enables us to extend \Cref{alg:pmatroids} to the 
multi-pass version based on McGregor's algorithm for MWM, and its 
$f$-extension, i.e., \Cref{alg:multipass}, and we get the approximation ratios 
$p+\eps$ and $p+1+\eps$ for modular and submodular case, respectively, by 
setting $\gamma=\eps/(p+1)$ for both linear and submodular case, and
$\kappa= \gamma^3/((p-1)(1+\gamma)^2-\gamma^3)$ for linear case and
$\kappa= \gamma^3/(p+(2p-1)\gamma +(p-1)\gamma^2-\gamma^3)$ for submodular 
case.

\section{Zelke's Algorithm for MWM}\label{app:zelke}

As mentioned earlier, Zelke's algorithm is \spcl. Therefore, to describe its
logic in full, it suffices to explain what happens in
\Crefrange{alg:generic:improve}{alg:generic:updateS}.

For each edge $uv\in M$, the algorithm
stores at most two shadow edges associated with each of $u$ and $v$, denoted by 
$\shadow{uv}{u}$ and $\shadow{uv}{v}$. When an edge $y_1y_2$ arrives in the 
stream, it considers the following set of at most seven edges in the vicinity 
of $y_1y_2$, viz., $T=\{y_1y_2, g_1y_1, a_1g_1, a_1c_1, g_2y_2, 
a_2g_2,a_2c_2\}$, where
\begin{itemize}
  \item $g_1y_1, g_2y_2$ are edges in $M$ that intersect with $y_1y_2$,
  \item $a_1g_1=\shadow{g_1y_1}{g_1}$ and $a_2g_2=\shadow{g_2y_2}{g_2}$,
  \item $a_1c_1$ and $a_2c_2$ are the edges in $M$ that intersect with $a_1g_1$ 
    and $a_2g_2$, respectively.
\end{itemize}
For a set $A \subseteq E$ and a matching $M \subseteq E$, let $M \intedges A$ 
denote the set of edges in $M$ that share a vertex with some edge in $A$. The 
algorithm picks
an \emph{augmenting set} $A\subseteq T$ that is a matching and maximizes 
the difference $w(A)-(1+\gamma)w(M\intedges A)$. If this difference is
positive, then $A$ is chosen in \Cref{alg:generic:improve}, i.e.,
the matching $M$ is updated as follows:
\[
  M \gets (M \setminus (M \intedges A)) \cup A\ .
\]
Then the set of shadow edges $S$ is updated as
\[
  S\gets \left(S\setminus \left(A\cup
  \bigcup_{uv\in M \intedges A}\{\shadow{uv}{u},\shadow{uv}{v}\}
  \right)\right)
  \cup (M\intedges A)\ ;
\]
note that since we are removing shadow edges of the edges in $M \intedges 
A$, the number of shadow edges remains at most $n$.
The edges in $M\intedges A$ then become shadow edges associated with the 
vertices that it shares with the edges in $M$. For example, if $A=\{a_1g_1\}$, 
then $M\intedges A=\{g_1y_1, a_1c_1\}$, and after updating $M$ and $S$ as 
mentioned above, $\shadow{a_1g_1}{g_1} =g_1y_1$, and $\shadow{a_1g_1}{a_1} 
=a_1c_1$.  This completes the description of Zelke's algorithm.